\documentclass[11pt]{article}
\usepackage{float}

\usepackage{framed}
\usepackage[linesnumbered,boxruled,vlined]{algorithm2e}
\usepackage{verbatim}
\usepackage{amsfonts,amsmath,amsthm}
\usepackage{enumerate}
\usepackage[font=small,labelfont=bf]{caption}
\usepackage{color}
\usepackage{graphicx}
\usepackage{tabulary}
\usepackage{diagbox}
\usepackage[dvipsnames]{xcolor}
\usepackage{pgfplots}
\usepackage[colorlinks=true, linkcolor=blue, urlcolor=blue, citecolor=ForestGreen]{hyperref}
\usepackage{comment}
\usepackage[utf8]{inputenc}
\usepackage[margin=1in]{geometry}

\usepackage{setspace}

\newcommand{\poly}{\operatorname{poly}}

\newcommand{\R}{\mathbb{R}}

\newcommand{\chamfer}{\operatorname{Chamfer}}

\newcommand{\nand}{\operatorname{NAND}}
\usepackage{footnote}

\newcommand{\cB}{\mathcal{B}}

\newcommand{\cF}{\mathcal{F}}

\newcommand{\cS}{\mathcal{S}}

\newcommand{\cV}{\mathcal{V}}

\newcommand{\eps}{\epsilon}

\newtheorem{theorem}{Theorem}
\newtheorem{lemma}{Lemma}

\newtheorem{proposition}{Proposition}
\newtheorem{definition}{Definition}

\usepackage[framemethod=TikZ]{mdframed}
\newcounter{Frame}

\title{Multi-Vector Embeddings are Provably More Expressive than Single Vector Embeddings}

\author{%
  Rajesh Jayaram \\
  Google Research \\
  \texttt{rkjayaram@google.com} \\
}

\begin{document}

\maketitle

\begin{abstract}
Multi-vector (MV) embeddings have become a powerful paradigm in neural information retrieval (IR), achieving high retrieval accuracy by representing data with multiple vectors and scoring them via the non-linear Chamfer similarity. Despite their widely perceived superiority over \textit{single-vector} (SV) embeddings which use inner product similarity, to date there is no formal proof that SV similarities cannot approximate MV similarities with the same representation size. Specifically, we ask the following: 
for any bounded dataset size $n \leq 2^{\poly(m)}$, what is the smallest dimension $D$ so that given any collection of MV embeddings $Q_1,\dots,Q_n,X_1,\dots,X_n \subset \R^d$ containing at most $m$ vectors each, there always exist $q_1,\dots,q_n$, $d_1,\dots,d_n \in \R^{D}$ satisfying $|\langle q_i, d_j \rangle - \chamfer(Q_i,X_j)| \leq \eps$ for all $i,j$?
    Recently, the MUVERA algorithm~\cite{jayaram2024muvera} demonstrated that $D =  m^{O(1/\eps^2)}$ is possible.  If improved to $D = md$, this would imply that MV embeddings are no more expressive than SV embeddings.

In this paper, we rule out this scenario. 
Specifically, we prove the existence of a collection of MV embeddings in $\R^d$, each containing at most $m$ vectors, which require single-vector dimension of $D =(\eps^2 m)^{\Omega(1/\eps)}$ to approximate, establishing a strong separation in representation size between MV and SV embeddings. 
Our proof leverages the Pattern Matrix Method by constructing a hard instance whose Chamfer similarity matrix encodes the $NAND_k$ boolean function. Our results confirm a long-held belief in the IR community: at a fixed representation size, multi-vector embeddings can express similarities which cannot even be approximately represented by single vector embeddings. This result provides strong theoretical justification for the continued development of multi-vector models in practice.   
\end{abstract}



\section{Introduction}

A fundamental problem in the field of Neural Information Retrieval (IR) is the representation of data by elements in a metric space~\cite{zhang2016neural}. The traditional approach has been to utilize single-vector (SV) embedding models, which produce one embedding vector $x \in \R^d$ per data-point---such as an image, video, or paragraph---and the similarity between two data-points is given by the inner product $\langle x,y \rangle$ of their embeddings. SV models have become extremely widespread over the past decade for IR tasks, in addition to many other tasks such as clustering and classification~\cite{muennighoff2022mteb}. 

More recently, \emph{multi-vector} (MV) representations, introduced in ColBERT~\cite{khattab2020colbert}, have been shown to deliver significantly improved performance on popular IR
benchmarks. MV Models produce a \textit{set} of embeddings, also called a point-cloud, for each data-point. Given a query MV representation $Q \subset \R^d$ and a document MV representation $X \subset \R^d$,  the query-document similarity is scored via the \emph{Chamfer Similarity}: defined as $\chamfer(Q,X) = \frac{1}{|Q|} \sum_{q \in Q} \max_{x \in X} \langle q,x\rangle$. This non-linear, asymmetric similarity measure allows for fine-grained interactions between the individual vectors in the query and document point-clouds.

 Since their introduction in \cite{khattab2020colbert}, multi-vector embedding models have received significant attention in both industry and academia. The reason for this popularity is largely predicated on the hypothesis that multi-vector embeddings are fundamentally more expressive than single-vector embeddings. This is supported experimentally, where multi-vector models significantly outperform single vector models on several benchmarks, especially multimodal benchmarks such as ViDoRE~\cite{faysse2025colpali}. However, despite this experimental success, there is to date no formal proof that multi-vector models are more expressive than single vector models at any fixed representation size.
 
 For instance, consider a collection of $n$ query MV embeddings $Q_1,\dots,Q_n\subset \R^d$ and $n$ document MV embeddings $X_1,\dots,X_n \subset \R^d$, each containing at most $m$ vectors; the representation size of each point-cloud is $md$, measured in words in the word RAM model, or floats in practice. If, hypothetically, there always exists corresponding single-vector embeddings $q_1,\dots,q_n,x_1,\dots,x_n \in \R^{d \cdot m}$ such that for every $i,j$ the inner product $\langle q_i, x_j \rangle$ is a high-accuracy approximation of the Chamfer Similarity $\chamfer(Q_i,X_j)$, then this would essentially \emph{disprove} the widely held view that MV models are inherently more expressive than SV models at a fixed representation size. Such a result would be a significant blow to the purported benefits of multi-vector, as it in effect implies that any multi-vector embedding can be represented just as well as a single-vector embedding \textit{with the same total number of dimensions}. 
 
This motivates the following question: is it possible to design two\footnote{Since Chamfer is asymmetric and inner products are symmetric, it is clear that any mapping from Chamfer into Inner product requires two \textit{distinct} functions $F_q,F_{doc}$, one for the query side and one for the document side. } functions $F_q,F_{doc}: (\R^d)^m \to \R^{D}$ such that, for any two sets $Q,X \subset \R^d$ of at most $m$ unit\footnote{When $Q,X$ contain vectors with norm at most $1$, then $\chamfer(Q,X) \in [-1,1]$, thus an $\eps$ additive approximation is the correct scaling for the problem.} vectors we have: 
 \[ \left| \langle F_q(Q), F_{doc}(X) \rangle- \chamfer(Q,X) \right|  \leq \eps  \]

Such a mapping can be thought of as an \textit{embedding} (in the sense of Metric \emph{Embedding} Theory~\cite{indyk20178}) of multi-vector similarity into inner-product similarity. The critical question is how small the target dimension $D$ must be. 
Recently, the authors of MUVERA~\cite{jayaram2024muvera} studied exactly this question, and demonstrated that such a mapping is possible with $D = m^{O(1/\epsilon^2)}\log^2(1/\delta)$, where $\delta$ is the failure probability.\footnote{The original dimensionality $d$ of the vectors in the point cloud disappears after a suitable application of the Johnson-Lindenstrauss transform\cite{johnson1984extensions}.} Since then, it has been open whether this dimensionality can be improved.  In particular, to demonstrate that multi-vector models are strictly more expressive than single-vector models, it would suffice to prove a lower bound of $D=\omega(md)$.


In this paper we partially resolve this open question by proving a lower bound of $m^{\Omega(1/\eps)}$ for the dimension of any single-vector approximation of Chamfer. If $d = \poly(m)$, which can be assumed WLOG whenever the number of query and document sets is at most $\exp(\poly(m))$ by applying the Johnson-Lindenstrauss transform, this demonstrates a strong separation between MV and SV embeddings. Namely, any $\eps =1/\poly(m)$ approximation of MV by SV must use \textit{exponentially} more dimensions, and for any constant $\eps$ it must
still use polynomially more dimensions. Specifically, our main theorem is as follows:

\begin{theorem}\label{thm:main}
 Fix any $\eps \in \left(0,\frac{1}{12}\right)$ and $m \geq 1/\eps^2$, and set $d=2m$. Then there exist unit vectors $x_1,\dots,x_{n_1} \in \R^d$ and sets $Y_1,\dots,Y_{n_2} \subset \R^d$ each containing at most $m$ unit vectors, such that any vectors $a_1,\dots,a_{n_1},b_1,\dots,b_{n_2} \in \R^{D}$ which satisfy
 \[    \left|\langle a_i, b_j \rangle- \chamfer(x_i,Y_j) \right| \leq \epsilon \]
 for all $i \in [n_1], j \in [n_2]$ must have dimension $D = (\eps^2 m)^{\Omega(1/\eps)}$.
    Moreover, we have $n_1 = (m/k)^k \cdot 2^k$ where $k=\Theta(1/\eps^2)$, and $n_2 = 2^m$. 
\end{theorem}

To put Theorem \ref{thm:main} into context, setting the failure probability $\delta < \frac{1}{\poly(n_1,n_2)}$, the MUVERA algorithm \cite{jayaram2024muvera} deterministically guarantees 
the existence of such vectors $a_i,b_j$ with dimension $D = m^{O(1/\eps^2)}$. Note that for $m > 1/\eps^{2+\gamma}$ for any constant $\gamma>0$, our lower bound becomes $m^{\Omega(1/\eps)}$. 
Thus, our lower bound nearly matches the upper bound of MUVERA, demonstrating that significant improvements for embedding MV into SV are not possible.  We also note that MUVERA is an oblivious embedding:  the mapping of point clouds to single-vectors does not depend on the set of point clouds to be embedded. On the other hand, Theorem \ref{thm:main} proves a lower bound against even  data-dependent (i.e. non-oblivious) embeddings, where the functions $F_q,F_{doc}$ are allowed to depend on the input point clouds $X_i,Y_j$.

Our proof utilizes the Pattern Matrix Method of Sherstov~\cite{sherstov2008pattern}, which is a method that allows one to lower bound the approximate rank of the so-called pattern matrix of a boolean function $f$ by the approximate polynomial degree of $f$. We construct the hard instance of Chamfer sets $x_i,Y_j$ such that their $n_1 \times n_2$ Chamfer similarity matrix $M$ is precisely equal to the pattern matrix of the $NAND_k$ function. By definition, the best $\eps$-approximation of these sets by single vectors is given by the $\eps$-approximate rank of $M$. We then lower bound the approximate polynomial degree of $NAND_k$ to complete the proof.

We observe that, in the hard instance in Theorem \ref{thm:main}, the queries are just single vectors $x_i \in \R^d$, and only the document side $Y_j \subset \R^d$ are point-clouds of size at most $m$. In this case, the Chamfer similarity reduces to $\max_{y \in Y} \langle x,y\rangle$, which is known as the Maximum Inner Product (MAX-IP) similarity. Thus, Theorem \ref{thm:main}  proves an even stronger result, by lower bounding the dimension of any approximation of the MAX-IP similarity by the Inner Product (IP) similarity.

A natural question is whether one can do better than MUVERA for embedding MAX-IP into IP. While we do not fully resolve this question, we demonstrate the existence of a \textit{Snowflake Embedding}~\cite{david1997fractured,tyson2005characterizations} of the Maximum Absolute Value of Inner Products Similarity (MAX-ABS-IP) into IP, defined via MAX-ABS-IP$(x,Y) = \max_{y \in Y}
|\langle x,y\rangle|$. Here, a snowflake embedding from a space $X$ to a space $Y$ means that similarities in $X$ are approximated by some positive real power of similarities in $Y$; in particular, such an embedding (approximately) preserves nearest neighbor orderings. Specifically, in Section \ref{sec:ub}, we prove the following theorem.

\begin{theorem}\label{thm:snowflake}
    Fix any error tolerance $\epsilon \in (0, 1)$, failure probability $\delta \in (0, 1)$, input dimension $d \geq 1$, and point cloud size $m \geq 1$. Let $k$ be the smallest even integer with $k \ge \frac{4}{\epsilon} \ln m$. Then there exist two oblivious, randomized mappings $F_q:\R^d \to \R^t$ and $F_{doc}: (\R^d)^m \to \R^t$,  where $t =  m^{O\left(\ln (1/\eps)/\epsilon\right)} \cdot \log(1/\delta)$, such that for any set of unit vectors $x, y_1, \dots, y_m \in \mathbb{R}^d$ we have
    \[ \left| | \langle F_q(x), F_{doc}(y_1,\dots,y_m) \rangle|^{1/k} -  \max_{i=1}^m |\langle x, y_i \rangle| \right| < \epsilon \]
with probability at least $1-\delta$.
\end{theorem}
In other words, Theorem \ref{thm:snowflake} demonstrates that an improved exponent of $O(\ln (1/\eps) / \eps)$ can be obtained for approximating the MAX-ABS-IP similarity with the $1/k$-th power of the IP similarity. While this comes closer to matching the exponent in the lower bound of Theorem \ref{thm:main}, it is important to note that Theorem \ref{thm:main} does not directly apply to the embedding of Theorem \ref{thm:snowflake}, since it would require $k=1$ and an approximation of the MAX-IP (not MAX-ABS-IP). We leave it as an open problem to determine if Theorem \ref{thm:snowflake} can be made to hold with $k=1$.

Despite applying to the seemingly weaker MAX-ABS-IP,  we note that Theorem \ref{thm:snowflake} can directly apply
to MAX-IP via a simple "non-negative lift" transformation. Specifically, given any original unit vectors $x, y \in \mathbb{R}^d$, define the affine transformation $\tilde{x} = \frac{1}{\sqrt{2}}(x \oplus 1)$ and $\tilde{y} = \frac{1}{\sqrt{2}}(y \oplus 1)$ in $\mathbb{R}^{d+1}$. This shifted space guarantees that all pairwise inner products are strictly non-negative, since $\langle \tilde{x}, \tilde{y} \rangle = \frac{1}{2}\langle x, y \rangle + \frac{1}{2} \in [0, 1]$,  yielding $\text{MAX-ABS-IP}(\tilde{x}, \tilde{Y}) = \frac{1}{2}\text{MAX-IP}(x, Y) + \frac{1}{2}$. Applying the Snowflake embedding to these transformed vectors yields single-vector representations whose inner product approximates $\left(\frac{1}{2}\text{MAX-IP}(x, Y) + \frac{1}{2}\right)^k$. 
Crucially, for nearest neighbor search one only needs to preserve the relative ordering of documents. Because the function $z \mapsto z^k$ is strictly monotonically increasing for $z \ge 0$, applying Theorem \ref{thm:snowflake} to the non-negative lifted vectors approximately preserves the rankings of the original MAX-IP similarity within the IP space. This demonstrates that single vector inner products are at least expressive enough to approximately preserve the \emph{ordering} of query-document similarities derived from a MAX-IP similarity model with dimension $m^{O(\frac{1}{\epsilon}\ln(1/\epsilon))}$, strictly improving upon the corresponding $m^{O(1/\epsilon^2)}$ bound resulting from MUVERA.  


Furthermore, we point out that for a broad class of modern representation models which naively produce non-negative embeddings, this geometric lifting is entirely unnecessary. For example, learned sparse retrieval models such as SPLADE~\cite{formal2021splade} explicitly apply a $\log(1 + \text{ReLU}(\cdot))$ activation over their vocabulary-weighting layers. In such ubiquitous settings, MAX-ABS-IP is identical to standard MAX-IP, and Theorem \ref{thm:snowflake} directly applies.





\subsection{Related work}

Since their introduction in \cite{khattab2020colbert}, there has been significant interest in the study of multi-vector embeddings, including modeling~\cite{gao2021coil,hofstatter2022introducing,lee2024rethinking,lin2024fine,qian2022multi,santhanam2021colbertv2,wang2021pseudo,yao2021filip, faysse2025colpali}, retrieval ~\cite{jayaram2024muvera, santhanam2022plaid, engels2024dessert,nardini2024efficient,scheerer2025warp}, and for fast algorithms for computing the Chamfer Similarity (and associated Chamfer Distance) itself \cite{bakshi2023near,halevi2026approximate,goranci2026fully}. However, to date there have only been a few papers which attempt to definitively study the expressive power of single vector as it relates to multi-vector embeddings. 

The first is the work of Wellner et al.~\cite{weller2025theoretical}, who investigated the theoretical limitations of SV embeddings, and show that for any fixed embedding dimension $d$, there are certain combinations of $k$ documents which cannot occur as the $k$-nearest neighbors of any query vector. The authors suggest that MV models may offer a way around this limit, but do not give any formal separation. Agarwal et al.~\cite{agarwal2026strengths} show that for a universe of random vectors on the unit sphere, and for a dataset of queries and documents which are subsets of this universe, single-vector models which average their set of vectors together perform worse in a certain ``goodness'' metric (introduced by the authors) when compared with MV models. This indicates that averaging is likely a poor mapping from MV into SV, but falls significantly short of proving a general lower bound against \textit{all} possible mappings. Finally, Lakshman et al.~\cite{lakshman2025breaking} show that MV retrieval has an improved ``stability" property when compared to average pooling; however, similarly their work does not provide explicit dimension bounds that would separate MV and SV embeddings. 

Thus, while the aforementioned works establish that single-vector embeddings struggle with complex retrieval geometries—and that multi-vector embeddings may overcome them, they do not discuss the question of explicit representation size, leaving open the possibility that a single-vector model with $O(md)$ dimensions could give an arbitrarily good approximation of a multi-vector model outputting sets of $m$ vectors in $d$ dimensions.

\section{Preliminaries}

As previously mentioned, given any sets $Q,X \subset \R^d$, we define:
\[\chamfer(Q,X) = \frac{1}{|Q|} \sum_{q \in Q} \max_{x \in X} \langle q,x\rangle\]
Following the convention in the IR literature, we will refer to the first operand of $\chamfer(\cdot ,\cdot)$ as the \textit{query}, and the second as the \textit{document}. We will often refer to any such subset of vectors with the same dimension, e.g. $A \subset \R^d$, as a \textit{point-cloud}. We say that a collection $\mathcal{A} = \{A_i\}_i$ of point clouds is $m$-bounded if $|A_i| \leq m$ for all $A_i \in \mathcal{A}$. For any $d \geq 1$ we let $\cB^d = \{x \in \R^d \; | \; \|x\|_2 \leq 1\}$ be the $d$-dimensional unit ball, and use the standard notation $\cS^{d-1} = \{x \in \R^d \; | \; \|x\|_2 = 1\}$ for the unit sphere in $\R^d$. Further, we define $\cF^d_m = \{ A \subset \R^d \; | \; |A| \leq m \}$ as the set of $m$-bounded point-clouds in $d$ dimensions, and similarly $\cB^d_m = \{ A \subset \cB^d \; | \; |A| \leq m \}$ and $\cS^{d-1}_m = \{ A \subset \cS^{d-1} \; | \; |A| = m \}$. 






\begin{definition}
    An $(\eps,\delta)$ oblivious embedding from the Chamfer similarity into the (Euclidean) inner product consists of two randomized mappings, $F_q: \cB^d_m\to \mathbb{R}^D$ and $F_{doc}: \cB^d_m \to \mathbb{R}^D$, such that for any $X,Y \in \cB^d_m $ :
\[ 
\mathbb{P}\Big( \big| \langle F_q(X), F_{doc}(Y) \rangle - \chamfer(X,Y)\big| \le \epsilon \Big) \ge 1- \delta
\]
Further, we can define an $(\eps,\delta)$ oblivious embedding from maximum inner-product into inner product as mappings $F_q: \cB^d\to \mathbb{R}^D$ and $F_{doc}: \cB^d_m \to \mathbb{R}^D$, such that for any $x \in \cB^d, Y \in \cB_m^d$:

\[ 
\mathbb{P}\Big( \big| \langle F_q(x), F_{doc}(Y) \rangle - \max_{y \in Y} \langle x, y \rangle \big| \le \epsilon \Big) \ge 1-\delta
\]
\end{definition}
By summing the query functions together, using linearity of the inner product and a union bound, it is easy to see that an $(\eps,\delta/m)$ oblivious embedding for MAX-IP into IP implies an $(\eps,\delta)$ oblivious embedding from Chamfer over sets of size $m$ into IP. 
In MUVERA \cite{jayaram2024muvera}, the authors proved the following:



\begin{theorem}[MUVERA\cite{jayaram2024muvera}]\label{thm:muvera}
   For any $\eps ,\delta  \in (0,1)$, there exists a randomized $(\eps,\delta)$ oblivious embedding from Chamfer into IP with dimension $D = m^{O(1/\eps^2)}\cdot \log^2(1/\delta)$.\footnote{The conference version of this paper claimed a target dimension $m^{O(1/\eps)}$ but contained an error in the proof. The fixed version with the corrected bound is given in \cite{dhulipala2024muveraarxiv}} 
\end{theorem}

It follows from Theorem \ref{thm:muvera} that for any set of at most $2^{\poly(m)}$ query and document point clouds containing vectors with norm at most $1$, there exists a set of single-vectors in dimension $m^{O(1/\eps^2})$  which approximate all pairwise Chamfer similarities of the point clouds to additive error $\eps$. Since $n_1,n_2 \leq 2^m$ in the hard instance of Theorem \ref{thm:main}, it follows that an upper bound of $D=m^{O(1/\eps^2})$ exists for our hard instance. 

We now introduce two key ingredients in our lower bound: the approximate polynomial degree of a function, and the approximate rank of a matrix. Firstly, for any function $f: \{0,1\}^n \to \mathbb{R}$ and $\eps \in (0,1)$, the $\epsilon$-approximate degree of $f$, denoted $\deg_\epsilon(f)$, is the minimum degree of a real polynomial $p$ such that $\lVert f - p \rVert_\infty \le \epsilon$. 
Secondly, for a matrix $F \in \mathbb{R}^{m \times n}$ and a parameter $\delta \ge 0$, the $\delta$-approximate rank of $F$, denoted $\operatorname{rk}_\delta(F)$, is defined as:
\[
\operatorname{rk}_\delta(F) = \min \{ \operatorname{rk}(A) : \lVert F - A \rVert_\infty \le \delta \}
\]

Our main result is a lower bound, showing that such a pair of mappings $(F_q,F_{doc})$ must satisfy $D = (\eps^2 m)^{\Omega(1/\eps)}$. Specifically, we prove the following. 

\begin{theorem}
\label{thm:main-intro}
    Fix any accuracy parameter $\eps \in \left(0,\frac{1}{12}\right)$ and point cloud size $m \geq \frac{1}{\eps^2}$, and let $d=2m$. Then there exists a set of unit vectors $\{x_i\}_{i \in [n_1]} \subset \cS^{d-1}$ and point clouds $\{Y_j\}_{j \in [n_2]} \subset \cS_m^{d-1}$ such that the matrix $M \in \R^{n_1 \times n_2}$ defined by $M_{i,j} = \max_{y \in Y_j} \langle x_i , y\rangle$ satisfies:
    \[\operatorname{rk}_\epsilon(M) \geq (\eps^2 m)^{\Omega(1/\eps)}\]
    where $n_1 = (m/k)^k \cdot 2^k$ for $k=\lfloor \frac{1}{64\eps^2}\rfloor$, and $n_2 = 2^m$. 
\end{theorem}
Since the approximate rank of $M$ is precisely equal to the minimum dimension $D$ of vectors $a_1,\dots,a_{n_1},b_1,\dots,b_{n_2}$ satisfying $\|AB^T - M\|_\infty \leq \eps$, where the rows of $A \in \R^{n_1 \times D},B \in \R^{n_2 \times D}$ are given by the vectors $a_i,b_j$ respectively, Theorem \ref{thm:main-intro} immediately implies Theorem \ref{thm:main}. Thus, in what follows we will prove Theorem \ref{thm:main-intro}.
 

\subsection{The Pattern Matrix Method}
Our proof of Theorem \ref{thm:main-intro} will employ the Pattern Matrix Method of Sherstov \cite{sherstov2008pattern}. Specifically, 
we proceed by constructing a hard instance of point clouds such that the query-document maximum inner product similarity matrix is precisely equal to the \textit{pattern matrix} of the boolean $\nand_k$ function. We begin by introducing the pattern matrix.

   Fix integers $k,n\geq1$ with $k<n$ and $k|n$, and partition $[n]$ into $k$ equally sized blocks $B_1,\dots,B_{k}$, where 
   \[B_i = \left\{(i-1) \cdot\frac{n}{k}+1 , \dots, i \cdot \frac{n}{k} \right\} \]
   Next, let $\cV(n,k)$ denote the family of subsets $V \subset [n]$ that have exactly one element from each of the blocks $B_i$. Clearly $|V| = k$ for each $ V \in \cV(n,k)$, and moreover $|\cV(n,k)| = (n/k)^k$. Next, for a bit string $x \in \{0,1\}^n$ and $V \in \cV(n,k)$, define the projection of $x$ onto $V$ by $x|_V = (x_{i_1}, x_{i_2}, \dots, x_{i_k})$,
   where $i_1 < i_2 < \dots < i_k$ are the elements of $V$. Given this, we can now define the pattern matrix.

   \begin{definition}[Pattern Matrix \cite{sherstov2008pattern}]
       Fix any function $f:\{0,1\}^k \to \R$. Then the $(n,k,f)$ Pattern matrix is the matrix $A$ given by:
       \[ A = \Big[ f\big(w \oplus x|_V \big) \Big]_{x \in \{0,1\}^n , (V,w) \in \cV(n,k) \times \{0,1\}^k}    \]
   \end{definition}
   Thus, the $(n,k,f)$ pattern matrix $A$ has $2^n$ rows and $2^k \cdot (n/k)^k$ columns, where the rows are indexed via bitstrings $x \in \{0,1\}^n$ and the columns are indexed by a pair $(V,w)$ where $V \in \cV(n,k)$ and $w \in \{0,1\}^k$.
One remarkable property of pattern matrices demonstrated in \cite{sherstov2008pattern} is that one can lower bound the \emph{approximate rank} of a $(n,k,f)$ pattern matrix as a function of the $\eps$-approximate degree of a boolean function $f$ mapping into $\{-1,1\}$. 

\begin{theorem}[Theorem 8.1 \cite{sherstov2008pattern}]
\label{thm:pattern-mat}
    Let $F$ be the $(n, k, f)$-pattern matrix, where $f: \{0,1\}^k \to \{-1, +1\}$ is given. Then for every $\epsilon \in [0,1)$ and every $\delta \in [0, \epsilon]$,
\[
\operatorname{rk}_\delta(F) \ge \left( \frac{\epsilon - \delta}{1 + \delta} \right)^2 \left(\frac{n}{k}\right)^{\deg_\epsilon(f)}
\]
\end{theorem}

\textbf{Roadmap:} We first construct a hard family of query vectors and document point clouds $\{x_i\}, \{Y_j\}$ such that their pair-wise Max-IP similarity matrix exactly realizes the $(m,k,f)$ pattern matrix, where $f = \nand_k$ is the NotAnd function over $k$ bits, and $k = \Theta(\frac{1}{\eps^2})$. Our main result will follow after lower bounding the approximate degree of $\nand_k$.

\section{Construction of the Hard Instance}

We now describe the construction of our hard instance.  
Fix any $\eps \in \left(0,\frac{1}{12}\right)$ and $m \geq \frac{1}{\eps^2}$, and set $k = \lfloor 1 / (64\epsilon^2) \rfloor$, noting that $ k \geq 2$. We can assume WLOG that $m$ is a multiple of $k$, since if it is not we can round $m$ down to the next multiple of $k$ in our construction without affecting the asymptotics of the main result. Specifically, we will begin by proving the following Lemma:

\begin{lemma}[Hard Instance Construction]
\label{lem:hard-instance}
Fix any $\epsilon \in \left(0,\frac{1}{12}\right)$ and set  $k = \lfloor 1/(64\epsilon^2) \rfloor$, and fix any integer $m \ge 1/\epsilon^2$ that is a multiple of $k$. There exists a set of $n_1 = 2^k(m/k)^k$ query unit vectors $\{x_i\}_{i=1}^{n_1} \subset \mathbb{R}^{2m}$ and a set of $n_2 = 2^m$ document point clouds $\{X_j\}_{j=1}^{n_2} \subset \mathbb{R}^{2m}$, each containing exactly $m$ unit vectors, such that the Chamfer similarity matrix $M \in \mathbb{R}^{n_1 \times n_2}$, defined by $M_{i,j} = \max_{p \in X_j} \langle x_i, p \rangle$, satisfies:
$$\sqrt{k} \cdot M^T = F_{\text{NAND}}$$
where $F_{\text{NAND}}$ is the $(m, k, \text{NAND}_k)$ pattern matrix and $ \nand_k: \{0,1\}^k \to \{0,1\}$ is the Not-And function on $k$ bits. 
\end{lemma}

\begin{proof}

We begin by following the setup of the pattern matrix: we partition $[m] = \{1,2,\dots,m\}$ into $k$ contiguous blocks $B_1, \dots, B_k$, each of size $m/k$, namely

 \[B_i = \left\{(i-1) \cdot\frac{m}{k}+1 , \dots, i \cdot \frac{m}{k} \right\} \]

 As in the earlier definition of the Pattern Matrix, we write $\cV(m,k)$ to denote the family of subsets $V \subset [m]$ that have exactly one element from each of the blocks $B_i$. As in the Lemma statement, we set the ambient dimension to be $d = 2m$, so that our query vectors and document point clouds will live in $\R^{2m}$. We proceed by defining the set of query and document point clouds.

 \textbf{Construction of Queries.}  A query $x_i$ is defined by a pair $(V,w)$ where $V \in \cV(m,k)$ and $w \in \{0,1\}^k$. Specifically, given any $V= (v_1, \dots, v_k)  \in \cV(m,k)$ and $w \in \{0, 1\}^k$, we define the vector:
    \[ 
    q_{V,w} = \frac{1}{\sqrt{k}} \sum_{j=1}^k \Big( w_j e_{v_j} + (1-w_j) e_{m+v_j} \Big) 
    \]

    where $e_i \in \R^{2m}$ is the $i$-th standard basis vector.
   Observe that $q_{V,w}$ has exactly $k$ non-zero entries, each whose value is $1/\sqrt{k}$, thus $\|q_{V,w}\|_2 = 1$. Moreover, there are exactly $n_1= 2^k (m/k)^k$ query vectors.

    \textbf{Construction of Documents.} A document set $P$ is defined by a vector $y \in \{0,1\}^m$. Specifically, for every boolean string $y \in \{0, 1\}^m$, define the set $P_y = \{ p_1(y), \dots, p_m(y) \}$, where $p_i(y) = y_i e_i + (1-y_i) e_{m+i}$. Notice that each $p_i(y)$ is a standard basis vector, so $\|p_i\|_2 = 1$. Moreover, the number of vectors in the document point cloud is exactly $m$, and there are exactly $n_2 = 2^m$ document point clouds.

\textbf{The Similarity Matrix.} Now we analyze the form of the $n_1 \times n_2$ similarity matrix $M \in \R^{n_1 \times n_2}$. Fix any query vector $x_{(V,w)}$, given by $(V,w) \in \cV(m,k) \times \{0,1\}^k$, and fix any document point cloud $P_y$, given by $y \in \{0,1\}^m$. Now for any $i \in [m]$, note that the inner product between $x_{(V,w)}$ and the $i$-th vector $p_i(y) \in P_y$ is $0$ if $i \notin V$. Otherwise, if $i = v_j \in V$, then the inner product is precisely $\frac{1}{\sqrt{k}} [w_j y_{v_j} + (1-w_j)(1-y_{v_j})]$.  Observe that the bracketed term equals $1$ if $w_j = y_{v_j}$, and $0$ if $w_j \neq y_{v_j}$. Thus, in summary:

\[  \langle x_{(V,w)} ,p_i(y)  \rangle = \begin{cases}\
  \frac{1}{\sqrt{k}}  & \text{if }i = v_j \in V \text{ and } w_j = y_{v_j}\\
    0 & \text{otherwise} \\
\end{cases} \]
    Therefore, the Max-IP (and thus Chamfer) similarity evaluates to:
    \[ 
    M_{(V,w), y} = \max_{p \in P_y} \langle q_{V,w}, p \rangle = \frac{1}{\sqrt{k}} \max_{j=1}^k I(w_j = y_{v_j}) 
    \]
    Where $I(E) \in \{0,1\}$ is the indicator function for a boolean event $E$. 
    Observe that $\max_{j=1}^k I(w_j = y_{v_j}) = 1$ if and only if at least one index $j \in [k]$ satisfies  $w_j = y_{v_j}$, or equivalently if not all $j \in [k]$ satisfy  $w_j \neq y_{v_j}$. Thus, if we consider the vector $w \oplus y|_V \in \{0,1\}^k$, it follows that $\max_{j=1}^k I(w_j = y_{v_j}) = 1$ if and only if not all of the bits in $w \oplus y|_V$ are $1$ -- i.e. if the $\text{NAND}_k(w \oplus y|_V) = 1$, where $\text{NAND}_k$ is the $k$-bit Not-AND function, defined as $\text{NAND}_k(x) = 1-x_1\cdot x_2 \cdots x_k$ for a boolean vector $x \in \{0,1\}^k$.
    
    We summarize the above discussion as follows: 
    \begin{equation*}
        \begin{split}
          \sqrt{k}  M_{(V,w), y} = \max_{j=1}^k I(w_j = y_{v_j}) &= 1 - \min_{j=1}^k I(w_j \neq y_{v_j}) \\
            &= 1 - \min_{j=1}^k (w_j \oplus y_{v_j}) \\
            &= \text{NAND}_k(w \oplus y|_V)
        \end{split}
    \end{equation*}
    Thus, the Chamfer similarity matrix $M$ precisely encodes $\frac{1}{\sqrt{k}} \text{NAND}_k(w \oplus y|_V)$. So, by definition, $\sqrt{k} \cdot M^T$ is exactly equal to the Pattern Matrix of the $\text{NAND}_k:\{0,1\}^k \to \{0,1\}$ function, which completes the proof of the Lemma.
\end{proof}    
   Equipped with Lemma \ref{lem:hard-instance}, we would now like to apply Theorem \ref{thm:pattern-mat} to lower bound  the approximate rank of the pattern matrix $F_{NAND}$. However, Theorem \ref{thm:pattern-mat} requires that the function in question maps to $\{-1,1\}$ instead of $\{0,1\}$. So we instead consider the natural function $\text{NAND}_k^-:\{0,1\}^k \to \{-1,1\}$,  defined as $\text{NAND}_k^-(x) = 2\cdot\text{NAND}_k(x) - 1$. It follows that if $M$ is the Chamfer similarity matrix from Lemma \ref{lem:hard-instance}, then $(2 \sqrt{k} \cdot M^T - J)$ is precisely the $(m,k,\text{NAND}_k^-)$ pattern matrix, where $J \in \R^{n_2 \times n_1}$ is the all ones matrix. We observe that this transformation does not significantly affect the approximate rank of the matrix.

    \begin{proposition}\label{prop:approx-rank}
        Fix any matrix $M \in \{0,1/\sqrt{k}\}^{n_1 \times n_2}$ for any $k \geq 1$, and let $M' = 2 \sqrt{k} \cdot M^T - J$.  
        \[\operatorname{rk}_{\frac{1}{8\sqrt{k}}}(M) \geq \operatorname{rk}_{1/4}(M')-1 \]
    \end{proposition}
    \begin{proof}
        Fix any matrix $A \in \R^{n_1 \times n_2}$ with rank $\operatorname{rk}_{\frac{1}{8\sqrt{k}}}(M)$ such that $\|A-M\|_\infty \leq \frac{1}{8\sqrt{k}}$. Then clearly $\|2\sqrt{k}A^T-2\sqrt{k}M^T\|_\infty \leq \frac{1}{4}$, and moreover $\|(2\sqrt{k}A^T-J)-M'\|_\infty \leq \frac{1}{4}$. Further, $\operatorname{rank}(2\sqrt{k}A^T) = \operatorname{rank}(A)$, since scaling and transposing do not affect rank. Finally, since $J$ is a rank one matrix, by subadditivity of rank we have 
         \[ \operatorname{rk}_{1/4}(M') \leq \operatorname{rank}(2\sqrt{k}A^T-J) \leq \operatorname{rank}(A)+1=\operatorname{rk}_{\frac{1}{8\sqrt{k}}}(M) + 1\]
     which completes the proof.
    \end{proof}

By Proposition \ref{prop:approx-rank}, it will suffice to lower bound the $1/4$ approximate rank of the $(m,k,\text{NAND}_k^-)$ pattern matrix $F$. By Theorem \ref{thm:pattern-mat}, setting $\delta = 1/4$ and $\eps = 1/3$ we can lower bound this as follows:

\[\operatorname{rk}_{1/4}(F) \ge O(1) \cdot \left(\frac{m}{k}\right)^{\deg_{1/3}(\text{NAND}_k^-)}\]
Thus, it will suffice to lower bound the approximate degree \\$\deg_{1/3}(\text{NAND}_k^-)$, which we do in Section \ref{sec:approx-degree}. Specifically, we prove:

 \begin{lemma}
 \label{lem:approx-degree-main}
The $1/3$-approximate degree of $\text{NAND}_k^-$ is $\Omega(\sqrt{k})$.
\end{lemma}

Given this Lemma, we are now ready to complete the proof of our main theorem.
\begin{proof}[Proof of Theorem \ref{thm:main-intro}]
    By Lemma \ref{lem:hard-instance}, we have our desired hard instance $\{x_i\}_{i=1}^{n_1} \subset \mathbb{R}^{2m}$ and $\{X_j\}_{j=1}^{n_2} \subset \mathbb{R}^{2m}$ with the property that the Chamfer similarity matrix, defined by $M_{i,j} = \max_{p \in X_j} \langle x_i, p \rangle$,  satisfies:
$\sqrt{k} \cdot M^T = F_{\text{NAND}}$ where $F_{\text{NAND}}$ is the $(m, k, \text{NAND}_k)$ pattern matrix.  To prove the Theorem, recall that our final goal is to lower bound $\operatorname{rk}_{\epsilon}(M)$ by $(\eps^2 m)^{\Omega(1/\eps)}$.

As discussed earlier, it follows that the matrix $M'=2\sqrt{k}M^T-J$ is the $(m,k,\text{NAND}_k^-)$ pattern matrix. By Proposition \ref{prop:approx-rank}, using our setting of $k = \lfloor 1/(64 \eps^2) \rfloor$ and the fact that $\operatorname{rk}_{\theta}$ can only decrease as you increase the parameter $\theta$, we have:
\[ \operatorname{rk}_{\epsilon}(M) \geq \operatorname{rk}_{\frac{1}{8\sqrt{k}}}(M) \geq \operatorname{rk}_{1/4}(M')-1  \]
Finally, putting together Theorem \ref{thm:pattern-mat} with $\delta = 1/4$ and $\eps = 1/3$, and Lemma \ref{lem:approx-degree-main} we have

\[\operatorname{rk}_{1/4}(M') \ge O(1) \cdot \left(\frac{n}{k}\right)^{\deg_{1/3}(\text{NAND}_k^-)} \geq \left(\frac{m}{k} \right)^{(\Omega\sqrt{k})} = (\eps^2 m)^{\Omega(1/\eps)} \]
which completes the proof.
\end{proof}

\section{The Approximate Degree of $NAND_k^-$}
\label{sec:approx-degree}
The main goal of this section is to prove Lemma \ref{lem:approx-degree-main}. To do this, we will first show that the approximate degree of the $k$-input $NAND_k$ function is $\Omega(\sqrt{k})$, from which the same bound will follow for the $NAND_k^-$ function. Recall that for any function $f: \{0,1\}^n \to \mathbb{R}$, the $\epsilon$-approximate degree of $f$, denoted $\deg_\epsilon(f)$, is the minimum degree of any real polynomial $p$ such that $\lVert f - p \rVert_\infty \le \epsilon$.  We use a result from Paturi which characterizes the approximate degree of symmetric boolean functions \cite{Paturi_1992_stoc}. 

\begin{definition}[Symmetric Boolean Function]
A boolean function $f: \{0,1\}^n \rightarrow \{0,1\}$ is called symmetric if its value only depends on the number of $1$'s in the input. 
Formally, $f(x_1, \ldots, x_n) = f_i \in \{0, 1\}$ for all $(x_1, \ldots, x_n)$ such that $\sum_{j=1}^n x_j = i$. 
\end{definition}

\begin{lemma} \label{lemma:nand_symmetric}
The function $NAND_k: \{0,1\}^k \rightarrow \{0,1\}$ is a symmetric boolean function.
\end{lemma}
\begin{proof}
By definition, a $k$-input NAND gate outputs $0$ if and only if all $k$ of its inputs are $1$, otherwise it outputs $1$.  Thus, the output is $0$ if and only if $\sum_{j=1}^k x_j=k$, so the function is symmetric. 
\end{proof}

To determine the approximate degree of a symmetric boolean function, Paturi defines a quantity $\Gamma(f)$ that measures the location of the ``jump'' in the function's value \cite{Paturi_1992_stoc}. 

\begin{definition}[Jump Location $\Gamma(f)$]
For a symmetric boolean function $f:\{0,1\}^n \to \{0,1\}$ on $n$ variables, where $f_i$ is the value of $f$ on inputs with exactly $i$ $1$'s, $\Gamma(f)$ is defined as:
\begin{equation}
\Gamma(f) = \min\{|2i - n + 1| : f_i \neq f_{i+1} \text{ and } 0 \le i \le n - 1\} %
\end{equation}
\end{definition}

Paturi's theorem relates this quantity to the approximate degree:

\begin{theorem}[Paturi 1992 \cite{Paturi_1992_stoc}] \label{thm:paturi}
Let $f:\{0,1\}^n \to \{0,1\}$ be any boolean valued, non-constant symmetric function on $n$ variables, and fix any constant $c \in (0,1/2)$. Then the $c$-approximate degree of $f$ is $\Theta(\sqrt{n(n - \Gamma(f))})$. 
\end{theorem}

We cannot directly apply this result to $NAND_k^-$ since Paturi's theorem applies only to boolean functions mapping into $\{0,1\}$. Thus, we will instead first apply this theorem to $NAND_k$.

\begin{theorem}\label{thm:nand_k}
The $1/3$-approximate degree of $NAND_k$ is $\Theta(\sqrt{k})$.
\end{theorem}
\begin{proof}

By Lemma \ref{lemma:nand_symmetric}, $f = NAND_k:\{0,1\}^k \to \{0,1\}$ is a non-constant symmetric boolean function on $n = k$ variables. To calculate $\Gamma(NAND_k)$, we must find the indices $i$ where the value of the function changes, i.e., $f_i \neq f_{i+1}$.
As established in Lemma \ref{lemma:nand_symmetric}: $f_i = 1$ for $0 \le i \le k-1$, and $f_k=0$. Thus, $f_i \neq f_{i+1}$ only occurs when $i = k-1$.
Since there is only one such jump, the minimum over all jumps is simply the value at this index. Substituting $n = k$ and $i = k - 1$ into the expression for $\Gamma(f)$:
\[\Gamma(NAND_k) = |2(k - 1) - k + 1|= |2k - 2 - k + 1| = k - 1 \]
Assuming $k \ge 1$, we have $\Gamma(NAND_k) = k - 1$.
Now, we apply Theorem \ref{thm:paturi} with $n = k$ and $\Gamma(f) = k - 1$:
\begin{align*}
\text{deg}_{1/3}(NAND_k) &= \Theta\left(\sqrt{k(k - \Gamma(NAND_k))}\right) \\
&=  \Theta\left(\sqrt{k(k - (k - 1))}\right) \\
&= \Theta(\sqrt{k})
\end{align*}

\end{proof}


  
  We are now prepared to prove Lemma \ref{lem:approx-degree-main}.
  
\textbf{Lemma \ref{lem:approx-degree-main}} {\it
The $1/3$-approximate degree of $\text{NAND}_k^-$ is $\Omega(\sqrt{k})$.}
\begin{proof}

Recall $\text{NAND}_k^- : \{0,1\}^k \to \{-1,1\}$ is defined by $\text{NAND}_k^- = 2\cdot \text{NAND}_k -1$.
Let $p$ be any real polynomial with minimal degree that satisfies $\|p - \text{NAND}_k^-\|_\infty \le 1/3$. Letting $q = (p+1)/2$,  we have

\begin{align*} 
\|q - \text{NAND}_k\|_\infty &= \left\| \frac{p + 1}{2} - \frac{\text{NAND}_k^- + 1}{2} \right\|_\infty \\
&= \frac{1}{2} \|p - \text{NAND}_k^-\|_\infty \leq \frac{1}{6} < \frac{1}{3}\\
\end{align*}
Moreover, note that $q$ has the same degree as $p$.
By Theorem \ref{thm:nand_k}, any polynomial approximating $\text{NAND}_k$ to error $1/3$ must have degree at least $\Omega(\sqrt{k})$. Therefore:
\[ 
\deg_{1/3}(\text{NAND}_k^-) = \deg(p) = \deg(q) \ge \deg_{1/3}(\text{NAND}_k) = \Omega(\sqrt{k}) 
\]

\end{proof}

\section{Snowflake Embedding of MAX-ABS-IP into IP} \label{sec:ub}
In this section we prove Theorem \ref{thm:snowflake}. Namely, we give an oblivious \textit{Snowflake} Embedding of the Maximum Absolute Value of Inner-Product (MAX-ABS-IP) similarity into the inner product similarity. Recall that given a vector $x \in \R^d$ and a set $Y \subset \R^d$, the MAX-ABS-IP between $x$ and $Y$ is defined as $\max_{y \in Y} |\langle x,y\rangle|$. Here, ``Snowflake'' refers to the fact that the similarity will not be preserved at the exact same scale, and instead the output of the embedding will approximate the original similarity \emph{raise to the $k$-power} for some $k \geq 0$.
Before proving the main result, we first state a few standard results which will be useful for us. 


\begin{lemma}[JL for Inner Products~\cite{arriaga2006algorithmic}]
\label{lem:jl_inner}
Let $\epsilon_0 \in (0, 1)$ and $\delta \in (0, 1)$. For any fixed vectors $u, v \in \mathbb{R}^D$, let $S \in \mathbb{R}^{t \times D}$ be a random matrix with entries drawn uniformly at random from $\{-\frac{1}{\sqrt{t}}, \frac{1}{\sqrt{t}}\}$. If $t = \Theta\left(\epsilon_0^{-2} \log(1/\delta)\right)$, then with probability at least $1 - \delta$:
\[ \left| \langle Su, Sv \rangle - \langle u, v \rangle \right| \le \epsilon_0 \|u\|_2 \|v\|_2 \]
\end{lemma}

\begin{lemma}[Subadditivity of Fractional Roots]
\label{lem:subadditivity}
For any $A, B \ge 0$ and $k \ge 1$, we have $|A^{1/k} - B^{1/k}| \le |A - B|^{1/k}$.
\end{lemma}
\begin{proof}
Without loss of generality, assume $A \ge B \ge 0$. We want to show that $A^{1/k} - B^{1/k} \le (A-B)^{1/k}$, which rearranges to $A^{1/k} \le (A-B)^{1/k} + B^{1/k}$. 
Let $x = A-B \ge 0$ and $y = B \ge 0$. We must show $(x+y)^{1/k} \le x^{1/k} + y^{1/k}$, which is a standard factor for exponents $p = 1/k \le 1$. To see it, note if $x+y=0$, it holds trivially. Otherwise, if $x+y > 0$, dividing the RHS by $(x+y)^p$ gives $\left(\frac{x}{x+y}\right)^p + \left(\frac{y}{x+y}\right)^p$. Since each fractions is in $[0, 1]$ and $p \le 1$, we have $z^p \ge z$ for $z \in [0, 1]$. Thus, the sum is bounded below by $\frac{x}{x+y} + \frac{y}{x+y} = 1$ as desired.
\end{proof}


We now state the main Lemma describing our Snowflake embedding.
\begin{lemma}[Snowflake Embedding of MAX-ABS-IP into IP]
Let $x, y_1, \dots, y_m \in \mathbb{R}^d$ be any unit vectors. Let $\epsilon \in (0, 1)$ be an error tolerance parameter and $\delta \in (0, 1)$ be a failure probability. 
Let $k$ be the smallest even integer with $k \ge \frac{4}{\epsilon} \ln m$, and define $x' = x^{\otimes k}$ and $y' = \sum_{i=1}^m y_i^{\otimes k}$ which are vectors in $\mathbb{R}^{d^k}$. 

Let $S \in \mathbb{R}^{t \times d^k}$ be a random Johnson-Lindenstrauss projection matrix mapping to $t$ dimensions, where $t = \Theta\left( m^2 \left(\frac{2}{\epsilon}\right)^{2k} \log(1/\delta) \right)$.
Then, with probability at least $1 - \delta$, we have:
\[ \left| |\langle Sx', Sy' \rangle|^{1/k} -  \max_{i=1}^m |\langle x, y_i \rangle| \right| < \epsilon \]
\end{lemma}

\begin{proof}
Define $M = \max_{i=1}^m |\langle x, y_i \rangle|$.
Let $V = \langle x', y' \rangle$ denote the exact inner product in the high-dimensional tensor space $\mathbb{R}^{d^k}$, and let $E = \langle Sx', Sy' \rangle$ denote the approximated inner product after the dimension reduction. By the triangle inequality, we bound the total approximation error by splitting it into two parts:
\begin{equation}
\label{eq:triangle_main}
\left| |E|^{1/k} - M \right| \le \left| |E|^{1/k} - V^{1/k} \right| + \left| V^{1/k} - M \right|
\end{equation}
(We show momentarily in Step 2 that $V \ge 0$, which is why $V^{1/k}$ requires no absolute value). We allocate an error budget of $\epsilon / 2$ to each term.

\paragraph{Step 1: Johnson-Lindenstrauss Projection Bound.}
We first evaluate the distortion due to the projection $S$. We determine the $\ell_2$ norms of the high-dimensional constructed vectors $x'$ and $y'$:
\begin{itemize}
    \item By standard properties of tensor products, $\langle u^{\otimes k}, v^{\otimes k} \rangle = \langle u, v \rangle^k$. Since $x$ is a unit vector, we have $\|x'\|_2 = \sqrt{\langle x, x \rangle^k} = 1^{k/2} = 1$.
    \item By the triangle inequality, we have $\|y'\|_2 = \left\| \sum_{i=1}^m y_i^{\otimes k} \right\|_2 \le \sum_{i=1}^m \|y_i^{\otimes k}\|_2 = \sum_{i=1}^m 1 = m$.
\end{itemize}
Let $\epsilon_0 = \frac{1}{m} \left(\frac{\epsilon}{2}\right)^k$ be an intermediate target tolerance. By Lemma~\ref{lem:jl_inner}, projecting to $t = \Theta\left(\epsilon_0^{-2} \log(1/\delta)\right)$ dimensions guarantees that with probability at least $1 - \delta$:
\[ |E - V| \le \epsilon_0 \|x'\|_2 \|y'\|_2 \le \epsilon_0 (1)(m) = \left(\frac{\epsilon}{2}\right)^k \]
By the reverse triangle inequality and since $V \ge 0$, we have $\left| |E| - V \right| = \left| |E| - |V| \right| \le |E - V|$. Using the subadditivity of the $k$-th root from Lemma~\ref{lem:subadditivity}, we scale the additive error back down:
\begin{equation}
\label{eq:jl_err}
\left| |E|^{1/k} - V^{1/k} \right| \le \left| |E| - V \right|^{1/k} \le |E - V|^{1/k} \le \left( \left(\frac{\epsilon}{2}\right)^k \right)^{1/k} = \frac{\epsilon}{2}
\end{equation}
Substituting our $\epsilon_0$ requirement into the dimension formula yields the target dimension $t$:
\[ t = \Theta\left( \frac{1}{\epsilon_0^2} \log(1/\delta) \right) = \Theta\left( \frac{m^2}{(\epsilon/2)^{2k}} \log(1/\delta) \right) = \Theta\left( m^2 \left(\frac{2}{\epsilon}\right)^{2k} \log(1/\delta) \right) \]

\paragraph{Step 2: The Tensor Product Argument and $\ell_k$ Bound.}
Next, we expand the exact high-dimensional inner product $V$ to bound the theoretical tensor approximation error. Using the multilinearity of the tensor product:
\[ V = \left\langle x^{\otimes k}, \sum_{i=1}^m y_i^{\otimes k} \right\rangle = \sum_{i=1}^m \langle x^{\otimes k}, y_i^{\otimes k} \rangle = \sum_{i=1}^m \langle x, y_i \rangle^k \]
Crucially, because we constrained $k$ to be an strictly \textit{even} integer, the sign of each term neutralizes, meaning $\langle x, y_i \rangle^k = |\langle x, y_i \rangle|^k \ge 0$. This confirms $V \ge 0$. Thus:
\[ V^{1/k} = \left( \sum_{i=1}^m |\langle x, y_i \rangle|^k \right)^{1/k} \]
Notice $V^{1/k}$ evaluates identically to the $\ell_k$ norm of a vector $z \in \mathbb{R}^m$ given by $z_i = |\langle x, y_i \rangle|$. The maximum absolute inner product $M$ is exactly its $\ell_\infty$ norm $\|z\|_\infty$. Using the standard $p$-norm inequalities for $m$-dimensional vectors, $\|z\|_\infty \le \|z\|_k \le m^{1/k} \|z\|_\infty$, we bracket the estimate:
\[ M \le V^{1/k} \le m^{1/k} M \]
Because $x$ and $y_i$ are unit vectors, Cauchy-Schwarz forces $M \le 1$. Hence, the absolute approximation error is:
\[ \left| V^{1/k} - M \right| = V^{1/k} - M \le M(m^{1/k} - 1) \le m^{1/k} - 1 \]
We specified $k \ge \frac{4}{\epsilon} \ln m$, which implies $\frac{1}{k} \ln m \le \frac{\epsilon}{4}$. Exponentiating both sides yields $m^{1/k} \le e^{\epsilon/4}$. For $u \in (0, 1/4]$, the exponential function exhibits the strict bound $e^u < 1 + 2u$. Since $\epsilon \in (0, 1)$, our parameter $u = \epsilon/4 < 1/4$, granting us:
\begin{equation}
\label{eq:tensor_err}
\left| V^{1/k} - M \right| \le e^{\epsilon/4} - 1 < 2\left(\frac{\epsilon}{4}\right) = \frac{\epsilon}{2}
\end{equation}

\paragraph{Step 3: Conclusion.}
We combine our constituent bounds by substituting Equations~\eqref{eq:jl_err} and~\eqref{eq:tensor_err} back into the original triangle inequality (Equation~\eqref{eq:triangle_main}):
\[ \left| |\langle Sx', Sy' \rangle|^{1/k} - M \right| < \frac{\epsilon}{2} + \frac{\epsilon}{2} = \epsilon \]
The strict upper bound $\epsilon$ holds with probability at least $1 - \delta$, completing the proof.
\end{proof}

\section{Conclusion}
In this work, we give the first formal separation of multi-vector and single-vector embeddings. Specifically, we demonstrate the existence of a set of $n_1,n_2 \leq 2^m$ query and document multi-vector representations of size at most $m$, such that any single-vector representation giving a pair-wise $\eps$ approximation of all Chamfer similarities requires dimension $m^{\Omega(1/\eps)}$. This gives provable justification for the widely-held view that multi-vector models are more powerful than single-vector models. Moreover, this is a step closer to matching the upper bound of $m^{O(1/\eps^2)}$ given by MUVERA~\cite{jayaram2024muvera}. We leave it as an open problem to close the gap and determine precisely what the correct $\epsilon$ dependency should be. As a partial step towards this problem, we pose the open problem of determining the exact complexity of embedding MAX-IP into IP. Since our lower bound holds even for the less general MAX-IP similarity, it may possibly be easier to match it with an improved upper bound for MAX-IP. We believe that our Snowflake Embedding from Theorem \ref{thm:snowflake} gives a positive indication that such an improved embedding from MAX-IP into IP may exist.

\section*{Acknowledgments}
The author would like to thank Josh Alman for suggesting to use the Pattern Matrix Method as a means of proving lower bounds against the approximate rank of a Chamfer similarity matrix. 

\bibliographystyle{alpha}
\bibliography{sample, main}

@inproceedings{khattab2020colbert,
  title={Colbert: Efficient and effective passage search via contextualized late interaction over bert},
  author={Khattab, Omar and Zaharia, Matei},
  booktitle={Proceedings of the 43rd International ACM SIGIR conference on research and development in Information Retrieval},
  pages={39--48},
  year={2020}
}

@article{lakshman2025breaking,
  title={Breaking the Curse of Dimensionality: On the Stability of Modern Vector Retrieval},
  author={Lakshman, Vihan and Munyampirwa, Blaise and Shun, Julian and Coleman, Benjamin},
  journal={arXiv preprint arXiv:2512.12458},
  year={2025}
}

@article{agarwal2026strengths,
  title={On Strengths and Limitations of Single-Vector Embeddings},
  author={Agarwal, Mihir and Garg, Ankit and Kayal, Neeraj and Shiragur, Kirankumar and others},
  journal={arXiv preprint arXiv:2603.29519},
  year={2026}
}

@article{weller2025theoretical,
  title={On the theoretical limitations of embedding-based retrieval},
  author={Weller, Orion and Boratko, Michael and Naim, Iftekhar and Lee, Jinhyuk},
  journal={arXiv preprint arXiv:2508.21038},
  year={2025}
}

@article{halevi2026approximate,
  title={Approximate Algorithms for Chamfer Distance Under Translation},
  author={Halevi, Gil and Zhang, Daniel and Zhang, Jason},
  journal={arXiv preprint arXiv:2605.25280},
  year={2026}
}

@inproceedings{scheerer2025warp,
  title={WARP: An efficient engine for multi-vector retrieval},
  author={Scheerer, Jan Luca and Zaharia, Matei and Potts, Christopher and Alonso, Gustavo and Khattab, Omar},
  booktitle={Proceedings of the 48th international ACM SIGIR conference on research and development in information retrieval},
  pages={2504--2512},
  year={2025}
}

@article{arriaga2006algorithmic,
  title={An algorithmic theory of learning: Robust concepts and random projection},
  author={Arriaga, Rosa I and Vempala, Santosh},
  journal={Machine learning},
  volume={63},
  number={2},
  pages={161--182},
  year={2006},
  publisher={Springer}
}

@article{engels2024dessert,
  title={DESSERT: An Efficient Algorithm for Vector Set Search with Vector Set Queries},
  author={Engels, Joshua and Coleman, Benjamin and Lakshman, Vihan and Shrivastava, Anshumali},
  journal={Advances in Neural Information Processing Systems},
  volume={36},
  year={2024}
}

@inproceedings{formal2021splade,
  title={Splade: Sparse lexical and expansion model for first stage ranking},
  author={Formal, Thibault and Piwowarski, Benjamin and Clinchant, St{\'e}phane},
  booktitle={Proceedings of the 44th International ACM SIGIR Conference on Research and Development in Information Retrieval},
  pages={2288--2292},
  year={2021}
}

@String{ARXIV            = {arXiv preprint arXiv:}}

@Misc{T,
  Title                    = {Twitter},
  HowPublished             = {Available as \texttt{https://twitter.com/}}
}

@article{santhanam2021colbertv2,
  title={ColBERTv2: effective and efficient retrieval via lightweight late interaction (2022)},
  author={Santhanam, Keshav and Khattab, Omar and Saad-Falcon, Jon and Potts, Christopher and Zaharia, Matei},
  journal={URL preprint arXiv:2112.01488},
  year={2021}
}

@inproceedings{hofstatter2022introducing,
  title={Introducing neural bag of whole-words with colberter: Contextualized late interactions using enhanced reduction},
  author={Hofst{\"a}tter, Sebastian and Khattab, Omar and Althammer, Sophia and Sertkan, Mete and Hanbury, Allan},
  booktitle={Proceedings of the 31st ACM International Conference on Information \& Knowledge Management},
  pages={737--747},
  year={2022}
}

@article{gao2021coil,
  title={COIL: Revisit exact lexical match in information retrieval with contextualized inverted list},
  author={Gao, Luyu and Dai, Zhuyun and Callan, Jamie},
  journal={arXiv preprint arXiv:2104.07186},
  year={2021}
}

@article{lee2024rethinking,
  title={Rethinking the role of token retrieval in multi-vector retrieval},
  author={Lee, Jinhyuk and Dai, Zhuyun and Duddu, Sai Meher Karthik and Lei, Tao and Naim, Iftekhar and Chang, Ming-Wei and Zhao, Vincent},
  journal={Advances in Neural Information Processing Systems},
  volume={36},
  year={2024}
}

@article{lin2024fine,
  title={Fine-grained late-interaction multi-modal retrieval for retrieval augmented visual question answering},
  author={Lin, Weizhe and Chen, Jinghong and Mei, Jingbiao and Coca, Alexandru and Byrne, Bill},
  journal={Advances in Neural Information Processing Systems},
  volume={36},
  year={2024}
}

@article{qian2022multi,
  title={Multi-vector retrieval as sparse alignment},
  author={Qian, Yujie and Lee, Jinhyuk and Duddu, Sai Meher Karthik and Dai, Zhuyun and Brahma, Siddhartha and Naim, Iftekhar and Lei, Tao and Zhao, Vincent Y},
  journal={arXiv preprint arXiv:2211.01267},
  year={2022}
}

@inproceedings{wang2021pseudo,
  title={Pseudo-relevance feedback for multiple representation dense retrieval},
  author={Wang, Xiao and Macdonald, Craig and Tonellotto, Nicola and Ounis, Iadh},
  booktitle={Proceedings of the 2021 ACM SIGIR International Conference on Theory of Information Retrieval},
  pages={297--306},
  year={2021}
}

@article{yao2021filip,
  title={Filip: Fine-grained interactive language-image pre-training},
  author={Yao, Lewei and Huang, Runhui and Hou, Lu and Lu, Guansong and Niu, Minzhe and Xu, Hang and Liang, Xiaodan and Li, Zhenguo and Jiang, Xin and Xu, Chunjing},
  journal={arXiv preprint arXiv:2111.07783},
  year={2021}
}

@article{muennighoff2022mteb,
  title={MTEB: Massive text embedding benchmark},
  author={Muennighoff, Niklas and Tazi, Nouamane and Magne, Lo{\"\i}c and Reimers, Nils},
  journal={arXiv preprint arXiv:2210.07316},
  year={2022}
}

@article{tyson2005characterizations,
  title={Characterizations of snowflake metric spaces},
  author={Tyson, Jeremy T and Wu, Jang-Mei},
  journal={Annales Fennici Mathematici},
  volume={30},
  number={2},
  pages={313--336},
  year={2005}
}

@book{david1997fractured,
  title={Fractured fractals and broken dreams: self-similar geometry through metric and measure},
  author={David, Guy and Semmes, Stephen},
  number={7},
  year={1997},
  publisher={Oxford University Press}
}

@incollection{indyk20178,
  title={8: low-distortion embeddings of finite metric spaces},
  author={Indyk, Piotr and Matou{\v{s}}ek, Ji{\v{r}}{\'\i} and Sidiropoulos, Anastasios},
  booktitle={Handbook of discrete and computational geometry},
  pages={211--231},
  year={2017},
  publisher={Chapman and Hall/CRC}
}

@article{johnson1984extensions,
  title={Extensions of Lipschitz mappings into a Hilbert space},
  author={Johnson, William B and Lindenstrauss, Joram and others},
  journal={Contemporary mathematics},
  volume={26},
  number={189-206},
  pages={1},
  year={1984}
}

@article{zhang2016neural,
  title={Neural information retrieval: A literature review},
  author={Zhang, Ye and Rahman, Md Mustafizur and Braylan, Alex and Dang, Brandon and Chang, Heng-Lu and Kim, Henna and McNamara, Quinten and Angert, Aaron and Banner, Edward and Khetan, Vivek and others},
  journal={arXiv preprint arXiv:1611.06792},
  year={2016}
}

@article{dhulipala2024muveraarxiv,
  title={Muvera: Multi-vector retrieval via fixed dimensional encodings},
  author={Dhulipala, Laxman and Hadian, Majid and Jayaram, Rajesh and Lee, Jason and Mirrokni, Vahab},
  journal={arXiv preprint arXiv:2405.19504},
  year={2024}
}

@inproceedings{Paturi_1992_stoc,
  title={On the degree of polynomials that approximate symmetric boolean functions (preliminary version)},
  author={Paturi, Ramamohan},
  booktitle={Proceedings of the twenty-fourth annual ACM symposium on Theory of computing},
  pages={468--474},
  year={1992}
}

@inproceedings{sherstov2008pattern,
  title={The pattern matrix method for lower bounds on quantum communication},
  author={Sherstov, Alexander A},
  booktitle={Proceedings of the fortieth annual ACM symposium on Theory of computing},
  pages={85--94},
  year={2008}
}

@inproceedings{santhanam2022plaid,
  title={{PLAID}: An efficient engine for late interaction retrieval},
  author={Santhanam, Keshav and Khattab, Omar and Saad-Falcon, Jon and Potts, Christopher and Zaharia, Matei},
  booktitle={Proceedings of the 31st ACM International Conference on Information \& Knowledge Management},
  pages={1747--1756},
  year={2022}
}

@article{goranci2026fully,
  title={Fully Dynamic Algorithms for Chamfer Distance},
  author={Goranci, Gramoz and Jiang, Shaofeng and Kiss, Peter and Szilagyi, Eva and Yang, Qiaoyuan},
  journal={Advances in Neural Information Processing Systems},
  volume={38},
  pages={36314--36341},
  year={2026}
}

@article{bakshi2023near,
  title={Near-linear time algorithm for the chamfer distance},
  author={Bakshi, Ainesh and Indyk, Piotr and Jayaram, Rajesh and Silwal, Sandeep and Waingarten, Erik},
  journal={Advances in Neural Information Processing Systems},
  volume={36},
  pages={66833--66844},
  year={2023}
}

@inproceedings{faysse2025colpali,
  title={Colpali: Efficient document retrieval with vision language models},
  author={Faysse, Manuel and Sibille, Hugues and Wu, Tony and Omrani, Bilel and Viaud, Gautier and Hudelot, C{\'e}line and Colombo, Pierre},
  booktitle={International Conference on Learning Representations},
  volume={2025},
  pages={61424--61449},
  year={2025}
}

@article{nardini2024efficient,
  title={Efficient Multi-Vector Dense Retrieval Using Bit Vectors},
  author={Nardini, Franco Maria and Rulli, Cosimo and Venturini, Rossano},
  journal={arXiv preprint arXiv:2404.02805},
  year={2024}
}

@article{jayaram2024muvera,
  title={{MUVERA}: Multi-vector retrieval via fixed dimensional encoding},
  author={Jayaram, Rajesh and Dhulipala, Laxman and Hadian, Majid and Lee, Jason D and Mirrokni, Vahab},
  journal={Advances in Neural Information Processing Systems},
  volume={37},
  pages={101042--101073},
  year={2024}
}

\end{document}